\documentclass[letterpaper]{article} 
\usepackage{aaai25}  
\usepackage{times}  
\usepackage{helvet}  
\usepackage{courier}  
\usepackage[hyphens]{url}  
\usepackage{graphicx} 
\usepackage{natbib}  
\usepackage{caption} 
\usepackage{xcolor}
\usepackage{xspace}
\usepackage{amsmath}
\usepackage{amssymb}
\usepackage{stmaryrd}
\usepackage{amsthm}
\usepackage{tikz}
\usepackage{ltlfonts}
\usepackage{algorithm}
\usepackage{algorithmic}
\usepackage{multirow}

\newtheorem{example}{Example}
\newtheorem{theorem}{Theorem}
\newtheorem{lemma}{Lemma}
\newtheorem{definition}{Definition}

\usepackage{newfloat}
\usepackage{listings}
\DeclareCaptionStyle{ruled}{labelfont=normalfont,labelsep=colon,strut=off} 
\floatstyle{ruled}
\newfloat{listing}{tb}{lst}{}
\floatname{listing}{Listing}
\title{Shield Synthesis for LTL Modulo Theories}
\author{
	Andoni Rodríguez \textsuperscript{\rm 1,\rm 2,*}, 
	Guy Amir \textsuperscript{\rm 3,*}, 
	Davide Corsi\textsuperscript{\rm 4}, 
	César Sánchez\textsuperscript{\rm 1},
	Guy Katz \textsuperscript{\rm 5}
}
\affiliations{
	\textsuperscript{\rm 1} IMDEA Software Institute, Spain \\
	\textsuperscript{\rm 2} Universidad Politécnica de Madrid, Spain \\
	\textsuperscript{\rm 3} Cornell University, USA\\
	\textsuperscript{\rm 4} University of California, Irvine, USA\\
	\textsuperscript{\rm 5} The Hebrew University of Jerusalem, Israel\\
	andoni.rodriguez@imdea.org, gda42@cornell.edu, dcorsi@uci.edu, 
	cesar.sanchez@imdea.org, guykatz@cs.huji.ac.il
}

\usepackage{bibentry}

\begin{document}

\maketitle
\let\thefootnote\relax\footnotetext{* Main co-authors.}

\newcommand{\Nat}{\mathbb{N}}
\newcommand{\KWD}[1]{\textit{#1}}
\newcommand{\topoly}{\KWD{topoly}}
\newcommand{\inside}{\KWD{inside}}
\newcommand{\outside}{\KWD{outside}}
\newcommand{\toconf}{\KWD{toconf}}
\newcommand{\True}{\KWD{true}}
\newcommand{\False}{\KWD{false}}
\newcommand{\BL}{\KWD{bl}}
\newcommand{\DefinedAs}{\,\stackrel{\text{def}}{=}\,}
\newcommand{\Confs}{\KWD{Confs}}

\newcommand{\todo}[1]{{\color{red}{\textbf{[TODO]} #1}}}

\newcommand{\andoni}[1]{\marginpar{\textcolor{red}{AuthorA: #1}}}
\newcommand{\cesar}[1]{\marginpar{\textcolor{red}{AuthorC: #1}}}
\newcommand{\davide}[1]{\marginpar{\textcolor{green}{AuthorD: #1}}}
\newcommand{\guy}[1]{\marginpar{\textcolor{orange}{AuthorGK: #1}}}
\newcommand{\guyamir}[1]{\marginpar{\textcolor{blue}{AuthorGA: #1}}}

\newcommand{\cesartodo}[1]{\todo[linecolor=blue,backgroundcolor=blue!25,bordercolor=blue]{Cesar: #1}}
\newcommand{\cesartext}[1]{\todo[inline,linecolor=blue,backgroundcolor=blue!25,bordercolor=blue]{Cesar: #1}}
\newcommand{\andonitodo}[1]{\todo[linecolor=red,backgroundcolor=red!25,bordercolor=red]{Andoni: #1}}
\newcommand{\andonitext}[1]{\todo[inline,linecolor=red,backgroundcolor=red!25,bordercolor=red]{Andoni: #1}}

\newcommand{\AP}{\ensuremath{\mathsf{AP}}\xspace}
\newcommand{\alphabet}{\mathrm{\Sigma}}
\newcommand{\DefOR}{\ensuremath{\hspace{0.2em}\big|\hspace{0.2em}}}

\newcommand{\ShT}{S_{\calT}}
\newcommand{\tupleof}[1]{\langle#1\rangle}

\newcommand{\Always}{\LTLsquare}
\newcommand{\Event}{\LTLdiamond} 
\newcommand{\Next}{\LTLcircle}
\newcommand{\LTLPrev}{\LTLcircleminus}
\newcommand{\PrevNoFirst}{\LTLcircletilde}
\newcommand{\HasAlwaysBeen}{\LTLsquareminus}
\newcommand{\Once}{\LTLdiamondminus}
\newcommand{\Since}{\mathbin{\mathcal{S}}}
\newcommand{\BackTo}{\mathbin{\mathcal{B}}}
\newcommand{\WeakPrev}{\LTLcircletilde}
\newcommand{\U}{\mathbin{\mathcal{U}}}
\newcommand{\R}{\mathbin{\mathcal{R}}}

\newcommand{\Not}{\mathclose{\neg}}
\renewcommand{\And}{\mathrel{\wedge}}
\newcommand{\Or}{\mathrel{\vee}}
\newcommand{\Impl}{\mathrel{\rightarrow}}
\newcommand{\Into}{\Impl}

\newcommand{\LTL}{\ensuremath{\textup{LTL}}\xspace}
\newcommand{\LTLt}{\ensuremath{\textup{LTL}_{\calT}}\xspace}
\newcommand{\LTLB}{\ensuremath{\textup{LTL}_{\mathbb{B}}}\xspace}
\newcommand{\LTLZ}{\ensuremath{\textup{LTL}_{\mathbb{Z}}}\xspace}

\newcommand{\configuration}{choice\xspace}
\newcommand{\configurations}{choices\xspace}
\newcommand{\Configuration}{Choice\xspace}
\newcommand{\Configurations}{Choices\xspace}

\newcommand{\reaction}{reaction\xspace}
\newcommand{\reactions}{reactions\xspace}
\newcommand{\xbar}{\overline{x}}
\newcommand{\ybar}{\overline{y}}

\newcommand{\Lit}{\ensuremath{\mathit{Lit}}\xspace}
\newcommand{\phiT}{\ensuremath{\varphi_{\mathcal{T}}}\xspace}
\newcommand{\phiB}{\ensuremath{\varphi_{\mathbb{B}}}\xspace}
\newcommand{\phiLegal}{\varphi^\textit{legal}}
\newcommand{\phiExtra}{\varphi^{\textit{extra}}}
\newcommand{\phiEx}{\ensuremath{\phiExtra}\xspace}
\newcommand{\xs}{\ensuremath{\overline{x}}\xspace}
\newcommand{\ys}{\ensuremath{\overline{y}}\xspace}
\newcommand{\zs}{\ensuremath{\overline{z}}\xspace}
\newcommand{\es}{\ensuremath{\overline{e}}\xspace}
\newcommand{\sss}{\ensuremath{\overline{s}}\xspace}
\newcommand{\vs}{\ensuremath{\overline{v}}\xspace}
\newcommand{\ws}{\ensuremath{\overline{w}}\xspace}
\newcommand{\Pt}{\ensuremath{\mathit{pt}}}
\newcommand{\Nt}{\ensuremath{\mathit{nt}}}
\newcommand{\PT}[1]{#1^p}
\newcommand{\NT}[1]{#1^a}

\newcommand{\mycal}[1]{\ensuremath{\mathcal{#1}}\xspace}
\newcommand{\calC}{\mycal{C}}
\newcommand{\react}{\ensuremath{\textit{react}}}
\newcommand{\calR}{\mycal{R}}
\newcommand{\calT}{\mycal{T}}
\newcommand{\calQ}{\mycal{Q}}
\newcommand{\VR}{\ensuremath{\textit{VR}}\xspace}
\newcommand{\bconf}{\ensuremath{\textit{bconf}}\xspace}

\newcommand{\VE}{\ensuremath{\overline{v}_e}}
\newcommand{\VS}{\ensuremath{\overline{v}_s}}
\newcommand{\LoopStop}{\text{loopStop}}
\newcommand{\GetReaction}{\text{getReact}}
\newcommand{\ValidityReduce}{\text{validityReduce}}
\newcommand{\ThZ}{\mathcal{T}_\mathbb{Z}}
\newcommand{\ThR}{\mathcal{T}_\mathbb{R}}

\newcommand{\GameT}{\mathcal{G}^{\calT}}
\newcommand{\GameB}{\mathcal{G}^{\mathbb{B}}}
\newcommand{\suchThat}{\textit{. }}

\newcommand{\qrprec}{\mathrel{\preceq}}

\newcommand{\qreact}{\ensuremath{\textit{qreact}}}
\newcommand{\Vars}{\mathit{Vars}}
\newcommand{\Bool}{\mathbb{B}}
\newcommand{\Part}{\rightharpoonup}



\newcommand{\rhoT}{\ensuremath{\rho_{\mathcal{T}}}\xspace}
\newcommand{\rhoB}{\ensuremath{\rho_{\mathbb{B}}}\xspace}
\newcommand{\innerLoop}{\ensuremath{\textit{inner}\_\textit{loop}}}
\newcommand{\true}{\top}

\newcommand{\mysubsection}[1]{\medskip\noindent\textbf{#1}}

\newcommand{\relu}{\text{ReLU}\xspace}

\newcommand{\sat}{\texttt{SAT}\xspace}
\newcommand{\unsat}{\texttt{UNSAT}\xspace}
\newcommand{\timeout}{\texttt{TIMEOUT}\xspace}
\newcommand{\fail}{\texttt{FAIL}\xspace}

\newcommand{\turnLeft}{\texttt{LEFT}\xspace}
\newcommand{\turnRight}{\texttt{RIGHT}\xspace}

\newcommand{\vxs}{\ensuremath{v_{\xs}}\xspace}
\newcommand{\vys}{\ensuremath{v_{\ys}}\xspace}
\newcommand{\ves}{\ensuremath{v_{\es}}\xspace}
\newcommand{\vsss}{\ensuremath{v_{\sss}}\xspace}
\newcommand{\dom}{\mathbb{D}}

\newcommand{\TurnLeft}{\ensuremath{\texttt{turn\_left}}\xspace}
\newcommand{\TurnRight}{\ensuremath{\texttt{turn\_right}}\xspace}
\newcommand{\Provider}{\textbf{provider}\xspace}
\newcommand{\Partitioner}{\textbf{partitioner}\xspace}
\newcommand{\Choicer}{\textbf{choicer}\xspace}
\newcommand{\CB}{\ensuremath{C_{\mathbb{B}}}\xspace}
\newcommand{\WB}{\ensuremath{W}\xspace}
\newcommand{\CT}{\ensuremath{C_{\calT}}\xspace}
\newcommand{\WT}{\ensuremath{\mathcal{W}_{\calT}}\xspace}

\newcommand{\MVR}{\KWD{MVR}\xspace}
\newcommand{\FVR}{\KWD{FVR}\xspace}
\newcommand{\Vals}{\KWD{val}}
\newcommand{\WR}{\KWD{WR}\xspace}
\newcommand{\COMPONENT}[1]{\textsf{#1}\xspace}
\newcommand{\find}{\COMPONENT{find}}
\newcommand{\getchoice}{\COMPONENT{getchoice}}
\newcommand{\Qnow}{\ensuremath{Q_{\textit{now}}}\xspace}

\begin{abstract}
In recent years, Machine Learning (ML) models 
have achieved remarkable success in various domains.
However, these models also tend to demonstrate unsafe behaviors, precluding their deployment in safety-critical systems. 
To cope with this issue, ample research focuses on developing methods that guarantee the safe behaviour of a given ML model. 
A prominent example is \emph{shielding} which incorporates an external component 
(a ``shield'') 
that blocks unwanted behavior.
Despite significant progress, shielding suffers from a main setback: it is  currently geared towards properties 
encoded solely in propositional logics (e.g.,  LTL) and is unsuitable for  richer logics. This, in turn, limits the widespread applicability of shielding in many real-world systems. 
%
%
In this work, we address this gap, and extend shielding to LTL modulo theories, by building upon recent advances in reactive synthesis modulo theories. This allowed us to develop a novel approach for generating shields conforming to complex safety specifications in these more expressive, logics.
%
We evaluated our shields and demonstrate their ability to handle rich data with temporal dynamics.
%
%
To the best of our knowledge, this is the first approach for synthesizing shields for such expressivity. 
%

\end{abstract}

\section{Introduction} \label{sec:intro}
%
%
%
Recently, DNN-based agents trained using Deep Reinforcement Learning (DRL) have been shown to successfully control
reactive systems of high complexity (e.g., \cite{MaFa20})
, such as robotic platforms.
However, despite their success, DRL controllers 
still suffer from various safety issues; e.g., small perturbations to their inputs, resulting either from noise or from a 
 malicious adversary, 
can cause even state-of-the-art agents to react unexpectedly (e.g., \cite{GoShSz14})
.
This issue raises severe concerns regarding the deployment of DRL-based agents in safety-critical reactive systems.

%
In order to cope with these DNN reliability concerns, the formal methods community has recently put forth various tools and techniques that rigorously ensure the safe behaviour of DNNs (e.g., \cite{KaBaDiJuKo17})
, and specifically, of DRL-controlled reactive systems (e.g., \cite{BaAmCoReKa23})
. 
One of the main approaches that is gaining popularity, is \emph{shielding}~\cite{BlKoKoWa15,AlBloEh18}, i.e., the incorporation of an external component (a ``shield'') that \emph{forces} an agent to behave safely according to a given specification. 
This specification $\varphi$ is usually expressed as a propositional formula, in which the atomic propositions represent the inputs ($I$) and outputs ($O$) of the system, controlled by the DNN in question. 
Once $\varphi$ is available, shielding seeks to guarantee that \emph{all} behaviors of the given system $D$ satisfy $\varphi$ through the means of a shield $S$: whenever the system encounters an input $I$ that triggers an erroneous output (i.e., $O : D(I)$ for which $\varphi(I,O)$ does not hold), $S$ corrects $O$ and replaces it with another action $O'$,  to ensure that $\varphi$(I, O') does hold. 
Thus, the combined system  $D \cdot S$  never violates $\varphi$. 
%
%
Shields are appealing for multiple reasons: 
%
they do not require ``white box'' access to $D$, 
%
a single shield $S$ can be used for multiple variants of $D$, 
%
it is usually computationally cheaper than static methods like DNN verification etc.
%
Moreover, shields are intuitive for practitioners, since they are synthesized based on the required $\varphi$.
%
%
%
However, despite significant progress, modern shielding methods still suffer from a main setback: 
they are only applicable to specifications in which the inputs/outputs are over Boolean atomic propositions.  
This allows users to encode only discrete specifications, typically in \emph{Linear Temporal Logic} (LTL).
Thus, as most real-world systems rely on rich data specifications, this 
precludes the use of shielding  in various such domains, such as continuous input spaces.

%

In this work, we address this gap and present a novel approach for shield synthesis that makes use of LTL modulo theories ($\LTLt$), 
where  Boolean propositions are extended to literals from a (multi-sorted) first-order theory $\calT$.
%
%
%
%
%
Concretely, we leverage Boolean abstraction methods~\cite{RoCe23a}, which transform $\LTLt$ specifications into equi-realizable pure (Boolean) LTL specifications.
We combine Boolean abstraction with 
reactive $\LTLt$ synthesis~\cite{RoCe24}, extending the common LTL shielding theory into a $\LTLt$ shielding theory.
Using $\LTLt$ shielding, we are able to  construct shields for more expressive specifications. This, in turn, allows us to override unwanted actions in a (possibly infinite) domain of $\calT$, and guarantee the safety of DNN-controlled systems in such complex scenarios. 
In summary, our contributions are: (1) developing two methods for shield synthesis over $\LTLt$ and presenting their proof of correctness; (2) an analysis of the impact of the Boolean abstractions in the precision of shields; (3) a formalization of how to construct  optimal shields using objective functions; and (4) an empirical evaluation that shows the applicability of our techniques. 

\section{Preliminaries} \label{sec:prelim}

\subsubsection{LTL and $\LTLt$.}

We start from LTL~\cite{Pn77,MaPn95}, which has the following syntax:
\[
  \varphi  ::= \true \DefOR a \DefOR \varphi \lor \varphi \DefOR \neg \varphi
  \DefOR \Next \varphi \DefOR \varphi \U\varphi, 
\]
where  $a\in\AP$ is an \emph{atomic proposition}, 
$\{\land,\neg\}$ are the common Boolean operators of \emph{conjunction} and
\emph{negation}, respectively, and
$\{\Next,\U\}$ are the \emph{next} and \emph{until} temporal
operators, respectively.
Additional temporal operators include $\calR$ (\emph{release}), $\Event$(\emph{finally}), and
 $\Always$ (\emph{always}), which can be derived from the syntax above.
Given a set of atomic propositions $\overline{a}$ we use $\Vals(\overline{a})$ for a set of possible valuations of variables in $\overline{a}$ (i.e. $\Vals(\overline{a})=2^{\overline{a}}$), and we use $v_{\overline{a}}$ to range over $\Vals(\overline{a})$.
We use $\Sigma=\Vals(\AP)$.
The semantics of LTL formulas associates traces $\sigma\in\Sigma^\omega$ with
\LTL fomulas (where $\sigma \models \top$ always holds, and $\Or$ and $\neg$ are standard):
\[
  \begin{array}{l@{\hspace{0.3em}}c@{\hspace{0.3em}}l}
    \sigma \models a & \text{iff } & a \in\sigma(0) \\
     \sigma \models \Next \varphi & \text{iff } & \sigma^1\models \varphi \\
     \sigma \models \varphi_1 \U \varphi_2 & \text{iff } & \text{for some } i\geq 0\;\; \sigma^i\models \varphi_2, \text{ and } \\
    && \;\;\;\;\;\text{for all } 
    \text{for all } 0\leq j<i, \sigma^j\models\varphi_1 \\
  \end{array}
\]
A safety formula $\varphi$ is such that for every failing trace $\sigma\not\models\varphi$ there is a finite prefix $u$ of $\sigma$, such that all $\sigma'$ extending $u$ also falsify $\varphi$ (i.e. $\sigma'\not\models\varphi$).

 The syntax of LTL modulo theory (\LTLt) replaces atoms $a$ by literals $l$ from some theory $\calT$.
 %
 Even though we use multi-sorted theories, for clarity of explanation we assume that $\calT$ has only one sort and use $\dom$ for the domain, the set that populates the sort of its variables.
For example, the domain of linear integer arithmetic $\ThZ$ is
$\mathbb{Z}$ and we denote this by $\dom(\ThZ)=\mathbb{Z}$.
Given an \LTLt formula $\varphi(\zs)$ with variables $\zs$ the semantics of \LTLt now associate traces $\sigma$ (where each letter is a valuation of $\overline{z}$, i.e., a mapping from $\overline{z}$ into $\dom$) with $\LTLt$ formulae.
The semantics of the Boolean and temporal operators are as in \LTL, and for literals:
\[
  \begin{array}{l@{\hspace{0.3em}}c@{\hspace{0.3em}}l@{\hspace{0em}}}
    \sigma \models l & \text{iff } & \text{the valuation $\sigma(0)$ of $\zs$ makes $l$ true according to $\calT$} \\
  \end{array}
\]

%
%
%
%
%
%

\subsubsection{The Synthesis Problem.}

Reactive LTL synthesis~\cite{Th08,PiPnSa06}
is the task of producing a system that satisfies a given LTL specification $\varphi$, where
atomic propositions in $\varphi$ are split into variables
controlled by the environment (``input variables'') and by the system (``output variables''), denoted by $\es$ and $\sss$, respectively.
Synthesis corresponds to a game where, in each turn, the
environment player produces values for the input propositions, and the system player
responds with values of the output propositions.
A play is an infinite sequence of turns, i.e., an infinite interaction of the system with the environment.
%
%
A strategy  for the system is a tuple $\CB: \tupleof{Q,q_0,\delta,o}$
where $Q$ is a finite set of states, $q_0\in Q$ is the inital state,
$\delta:Q\times \Vals(\es) \Into Q$ is the transition function and
$o:Q\times\Vals(\es)\Into\Vals(\sss)$ is the output function.
\CB is said to be \emph{winning}
for the system if all the possible plays played
according to the strategy satisfy the LTL formula $\phiB$.
In this paper we use ``strategy'' and ``controller'' interchangeably.

We also introduce the notion of \emph{winning region} (WR), which encompasses all possible winning moves for the system in safety formulae.
A winning region $\WR: \tupleof{Q,I,T}$ is a tuple where $Q$ is a finite set of states, $I\subseteq Q$ is a set of initial states and $T:Q\times \Vals(\es)\Into{}2^{(Q\times\Vals(\sss))}$ is the transition relation, which provides for a given state $q$ and input $\ves$, all the possible pairs of legal successor and output $(q',\vsss)$.
For a safety specification, every winning strategy is ``included'' into the WR (i.e., there is an embedding map).
%
%
%
%
LTL realizability is the decision problem of whether there is a winning strategy for the system 
(i.e., check if $\text{WR} \neq \emptyset$), 
while LTL synthesis is the computational problem of producing one. 



However,  in $\LTLt$ synthesis \cite{RoCe23b,RoCe24}, the specification is expressed in a richer logic where propositions are replaced by literals from some $\calT$.
In $\LTLt$ the (first order) variables in specification $\phiT$ are still split into those controlled
by the environment ($\xs$), and those controlled
by the system ($\ys$), where $\xs \cap \ys = \emptyset$.
We use $\phiT(\xs,\ys)$ to empasize that $\xs\cup\ys$ are all the
variables occurring in $\phiT$.
The alphabet is now $\Sigma_{\calT}=\Vals(\xs\cup\ys)$ (note that now valuations map a variable $x$ to $\dom(\calT)$).
We denote by $t[\xs \leftarrow \vxs]$, the substitution in $t$ of variables $\xs$ by values $\vxs$ (similarly for $t[\ys \leftarrow \vys]$), and also
 $t[\es \leftarrow \ves]$ and $t[\sss \leftarrow \vsss]$ for Boolean variables (propositions).
A trace $\pi$ is an infinite sequence of valuations in $\dom(\calT)$, which induces an
infinite sequence of Boolean values of the literals occurring in
$\phiT$ and, in turn, an evaluation of $\phiT$ using the semantics of the temporal operators.
For example, given $\psi = \Always(y<x)$ the trace $\pi$
$\{(x:2,y:6),(x:15,y:27)\ldots\}$ induces $\{(\False),(\True)\ldots\}$.
We use $\pi_x$ to denote the projection of $\pi$ to the values of only $x$ (resp. $\pi_y$ for $y$).
A strategy or controller for the system in $\calT$ is now a tuple $\CT: \tupleof{Q,q_0,\delta,o}$
where $Q$ and $q_0\in Q$ are as before, and
$\delta:Q\times \Vals(\xs) \Into Q$ and
$o:Q\times\Vals(\xs)\Into\Vals(\ys)$.




\subsubsection{Boolean Abstraction.}
Boolean abstraction~\cite{RoCe23a} transforms an
\LTLt specification $\phiT$ into an \LTL specification
$\phiB$ in the same temporal fragment (e.g., safety to safety) that preserves realizability, i.e., $\phiT$ and $\phiB$ are \emph{equi-realizable}.
Then, $\phiB$ can be fed into an off-the-shelf synthesis
engine, which generates a controller or a WR for realizable instances.
Boolean abstraction transforms $\phiT$, which contains literals $l_i$, into $\phiB = \phiT[l_i \leftarrow s_i] \wedge \phiEx$, where
$\sss=\{s_i|\text{for each } l_i\}$ is a set of fresh atomic propositions
controlled by the system---such that $s_i$ replaces $l_i$---and where
$\phiEx$ is an additional sub-formula that captures the dependencies
between the $\sss$ variables.
The formula $\phiEx$ also includes additional environment variables
$\es$ that encode that
the environment can leave the system with the power to choose certain
valuations of the variables $\sss$.

We often represent a valuation $\vsss$ of the Boolean variables $\sss$ (which map each variable in $\sss$ to $\True$ or $\False$)
as a \emph{choice} $c$ (an element of $2^{\sss}$), where
$s_i\in c$ means that $\vsss(s_i)=\True$.   
The characteristic formula $f_c(\xs,\ys)$ of a choice $c$ is
\(
  f_c=\bigwedge_{s_i\in c} l_i \And \bigwedge_{s_i\notin c}\neg l_i.
\)
We use $\calC$ for the set of choices, i.e., sets of sets of $\sss$.
A \emph{reaction} $r\subset\calC$ is a set of choices, which
characterizes the possible responses of the system as the result to a
move of the environment.
The \emph{characteristic formula} $f_r(\xs)$ of a reaction $r$ is:
$
  (\bigwedge_{c\in r} \exists\ys. f_c) \And (\bigwedge_{c\notin r} \forall\ys\neg f_c)
$.
We say that $r$ is a valid reaction whenever $\exists \xs.f_r(\xs)$ is valid.
Intuitively, $f_r$ states that for some $\vxs$ by the environment, the system can respond with $\vys$ making the literals in some choice $c\in r$ but
cannot respond with  $\vys$ making the literals in choices
$c\notin r$.
The set \VR of valid reactions partitions precisely the moves of the
environment in terms of the reaction power left to the system.
For each valid reaction $r$ there is a fresh environment variable
$e\in\es$ used in $\phiExtra$ to capture the move of the environment that chooses reaction $r$.
The formula $f_r(\xs)[\xs \leftarrow \vxs]$ is true if a given valuation $\vxs$ of $\xs$ is one move of the environment characterized by $r$.
The formula $\phiEx$ restricts the environment such that
exactly one of the variables in $\es$ is true, which forces that the environment
chooses precisely one valid reaction $r$ when the variable $e$ that corresponds to $r$ is true.

\begin{example} [Running example and abstraction] \label{ex:runningEx}
  In this paper, we address shields for general safety $\LTLt$, but for the sake 
  of simplicity in the presentation, we consider a simpler example.
  Let $\phiT = \square (R_0 \wedge R_1)$ where:
\[
  R_0 : (x<10) \shortrightarrow \Next(y>9) \hspace{3em}
  R_1 : (x \geq 10) \shortrightarrow (y \leq x),\]
where $\xs=\{x\}$ is controlled by the environment and $\ys=\{y\}$ by
the system.
In integer theory $\ThZ$ this specification is realizable (consider the
strategy to always play $y:10$) and the Boolean abstraction first introduces $s_0$ to
abstract $(x<10)$, $s_1$ to abstract $(y>9)$ and $s_2$ to abstract
$(y \leq x)$.
Then %
\( \phiB = \varphi'' \wedge \Always (\phiLegal \Into \phiExtra) \)
where
$\varphi'' = (s_0 \shortrightarrow \Next s_1) \wedge (\neg s_0
\shortrightarrow s_2)$ is a direct abstraction of $\phiT$.
Finally, $\phiExtra$ captures the dependencies between the abstracted
variables:
\newcommand{\NN}{\phantom{\neg}}
\begin{align*}
  \phiExtra:
  \begin{pmatrix}
    \begin{array}{lrcl}
      \phantom{\wedge}& \big(e_0 & \Into & \big( f_{c_1} \vee f_{c_2} \big) \\[0.27em]
      \wedge & \big(e_{0^+} & \Into & \big( f_{c_1} \vee f_{c_2} \vee f_{c_3} \big) \\[0.27em]
      \wedge & \big(\NN e_1 & \Into & \big( f_{c_4} \vee f_{c_5} \vee f_{c_6} \big)
    \end{array}
  \end{pmatrix},
\end{align*}
%
%
where $f_{c_1}=(s_0 \wedge s_1 \wedge \neg s_2)$, $f_{c_2}=(s_0 \wedge \neg s_1 \wedge s_2)$, $f_{c_3}=(s_0 \wedge \neg s_1 \wedge \neg s_2)$, $f_{c_4}=(\neg s_0 \wedge s_1 \wedge s_2)$, $f_{c_5}=(s_0 \wedge s_1 \wedge \neg s_2)$ and $f_{c_6}=(\neg s_0 \wedge \neg s_1 \wedge s_2)$ and where $c_0=\{s_0, s_1, s_2\}$, $c_1=\{s_0,s_1\}$, $c_2=\{s_0,s_2\}$, $c_3=\{s_0\}$, $c_4=\{s_1,s_2\}$, $c_5=\{s_1\}$, $c_6=\{s_1\}$ and $c_7=\emptyset$.
Also, $\es = \{e_0, e_{0^+}, e_1\}$ belong to the environment and represent
$(x<10)$, $(x<9)$ and $(x \geq 10)$, respectively.
$\phiLegal:(e_0 \And\neg e_1 \And\neg e_2)\Or(\neg e_0 \And e_1 \And\neg e_2)\Or(\neg e_0 \And\neg e_1 \And e_2)$
%
%
encodes that $\es$ characterizes a
partition of the (infinite) input valuations of the environment and that only one of the $\es$ is true in every move; e.g., the valuation
$\ves:\tupleof{e_0 : \False,e_0^+ : \True,e_1 : \False}$ of $\es$
corresponds to the choice of the environment where only $e_0^+$ is true (and we use $\ves:e_0^+$ for a shorter notation).
Sub-formulae such as $(\neg s_0 \wedge s_1 \wedge s_2)$ represent the
\textit{choices} of the system (in this case, $c=\{s_1,s_2\}$), that is, given a decision of the
environment (a valuation of $\es$ that makes exactly one variable $e \in \es$ true),
the system can \emph{react} with one of the choices $c$ in the
disjunction implied by $e$.
Note that $f_c=\neg (x<2)\wedge(y>1)\wedge (y\leq x)$.
Also, note that $c$ can be represented as $\vsss:\tupleof{s_0 : \False, s_1 : \True, s_2 : \True}$.
\end{example}



\section{$\LTLt$ Shield Computation}
\label{sec:theoryShields}


\subsubsection{Problem overview.}

In shielding 
,  at every step, the external design $D$
---e.g., a DRL sub-system---
produces an output $\vsss$ from
a given input $\ves$ provided by the environment.
The pair $(\ves,\vsss)$ is passed to the shield $S$ which decides, at every step, whether $\vsss$ proposed by $D$ is safe with respect to some specification $\varphi$. 
The combined system $D \cdot S$ is guaranteed to satisfy $\varphi$.
Using $\LTLt$ we can define richer properties than in propositional $\LTL$ (which has been previously used in shielding).
For instance, consider a classic $D \cdot S$ context in which the $D$ is a robotic navigation platform and $\varphi: \Always(\texttt{LEFT} \shortrightarrow \Next \neg \texttt{LEFT})$.
Then, if $D$ chooses to turn \texttt{LEFT} after a  \texttt{LEFT} action, the shield will consider this second action to be dangerous, and will override it with e.g., \texttt{RIGHT}.
%
%
If \LTLt is used instead of LTL, specifications can be more
sophisticated including, for example, numeric data. 

\noindent
\fbox{
\begin{minipage}{0.45\textwidth}
\emph{
Shielding modulo theories  is the problem of building a shield $\ShT$ from a specification $\phiT$ in which at least one of the input variables $\xs$ or one of the output variables $\ys$ are not Boolean. }
\end{minipage}
}

A shield $\ShT$ conforming to $\phiT$ will evaluate if a pair $(\xs : \vxs,\ys : \vys)$ violates $\phiT$ and propose an overriding output $\ys:\vys'$ if so.
This way, $D \cdot \ShT$ never violates $\phiT$.
Note that it is possible that $\phiT$ is unrealizable, which means that
some environment plays will inevitably lead to violations, and hence $\ShT$ cannot be constructed, because $\text{WR}=\emptyset$.

\begin{example} [Running example as shield] \label{ex:runningSh}
Recall $\phiT$ from Ex.~\ref{ex:runningEx}. 
Also, consider $D$ receives an input trace $\pi_x = \tupleof{15,15,7,5,10}$, and produces the output trace $\pi_y = \tupleof{6,5,13,16,11}$.
We can see $D$ violates $\phiT$ in the fifth step, since $(x<10)$ holds in fourth step (so $v_y$ has to be such that $(y : v_y>9)$ in the fourth step and $(v_y : \leq 10)$ must hold in the fifth step) but $(y : 11 \leq 10)$ is not true.
Instead, a $\ShT$ conforming to $\phiT$ would notice that $(x : 10,y : 11)$ violates $\phiT$ in this fifth step and would override $v_y$ with $v_y'$ to produce $(y' : 10)$, which is the only possible valuation of $y'$ that does not violate $\phiT$.
Note that $\ShT$ did not intervene in the remaining steps. Thus, $D \cdot \ShT$ satisfies $\phiT$ in the example.
    
\end{example}

We propose two different architectures for $\ShT$: one following a deterministic strategy and another one that is non-deterministic.
Both start from $\phiT$ and use Boolean abstraction as the core for the temporal information, but they vary on the way to detect erroneous outputs from $D$ and how to provide corrections.

\subsubsection{Shields as Controllers.} 

The first method leverages $\LTLt$ controller synthesis (see Fig.~\ref{fig:sts}(a)) together with a component to detect errors in $D$.
The process is as follows, starting from a specification $\phiT$:
\begin{enumerate}
\item Boolean abstraction~\cite{RoCe23a} transforms $\phiT$ into an equi-realizable
  \LTL  $\phiB$. 
\item A Boolean controller \CB is synthesized from $\phiB$ (using e.g., \cite{meyer18strix}). \CB receives Boolean inputs
  $\ves$ and provides Boolean outputs $\vsss$.
\item We synthesize a richer controller \CT that receives
  $\vxs$ and produces outputs $\vys'$ in $\dom(\calT)$. Note the apostrophe in $\vys'$, since $\vys$ is the output of $D$.
\item
  $\ShT$ receives the output $\vys$ provided by $D$ and checks if the pair $(\vxs,\vys)$
  violates $\phiT$: if there is a violation, it overrides $\vys$ with $\vys'$; otherwise, it permits $\vys$.
\end{enumerate}

We now elaborate on the last steps.
To construct \CT from \CB we use~\cite{RoCe23b} 
, where \CT is composed of three sub-components:
%
\begin{itemize}
\item A \Partitioner function, which computes Boolean inputs $\ves$ from the richer inputs $\vxs$ via partitioning, checking whether $f_r(\vxs)$ is true. Recall that every $\vxs$ is guaranteed to belong to exactly one reaction $r$.
  \item A \textbf{Boolean controller} \CB, which receives $\ves$ and provides Boolean outputs $\vsss$.
  Also recall that each Boolean variable in $\vsss$ corresponds exactly to a literal in $\phiT$ and that  $f_c$ is the characteristic formula of $c$.
\item A \Provider, which computes an output $\vys'$ from $f_c$ and $r$ such that $\vys'$ is a model
\footnote{See the extended version of this paper~\cite{RoAmCoSaKa24} for a 
for additional details.}
   of the formula
  $\exists \ys'. \big(f_r(\xs) \Into f_c(\ys',\xs)\big)[\xs\leftarrow\vxs]$.
  \end{itemize}

The composition of the \Partitioner, \CB and the \Provider implements \CT (see \cite{RoCe23b,RoCe24}).
For each of these components $g$, let us denote with $g(t_0)=t_1$ the application of $g$ with input $t_0$.
Then, the following holds:

\begin{lemma} \label{lemm:satisfiability}
Let $\phiT$ be an \LTLt specification, $\phiB$ its equi-realizable Boolean abstraction and \CB a controller for  $\phiB$. 
Let $q$ be a state of $\CB$ after processing inputs $\vxs^0\ldots \vxs^n$, and let $\vxs$ be the next input from the environment.
Let $\textup{\Partitioner}(\vxs)=\ves$ be the partition corresponding to $\vxs$, $o(q,\ves)=\vsss$  be the output of \CB and $c$ be the choice associated to $\vsss$.
Then, if $r$ is a valid reaction, the formula $\exists \ys. f_r(\vxs) \Into f_c(\ys,\vxs)$ is satisfiable in $\calT$, where $c \in r$.
\end{lemma}

%

In other words, \CT never has to face with moves from \CB that cannot
%
%
%
 %
   be mimicked by \Provider with appropriate values for $\vys$.
   Thus, a shield $\ShT$
  based on \CT (see Fig.~\ref{fig:sts}(a)) first detects whether a valuation pair $(\xs : \vxs, \ys : \vys)$ suggested by $D$ violates $\phiT$.
To do so, at each time-step, \CT obtains $\vsss$ by \CB, with associated $c$, and it checks whether $f_r(\vxs) \Into f_c(\vxs,\vys)$ is valid, that
is, whether the output proposed is equivalent to the move of \CB (in the sense that every literal $l_i$ from $(\vxs,\vys)$ has the same valuation  $\vsss$ given by \CB).
If the formula is valid, then $\vys$ is maintained.
Otherwise, the \Provider is invoked to compute a $\vys'$ that matches the move of \CB.
In both cases, $\ShT$ guarantees that $\vsss$ remains unaltered in \CB and therefore $\phiT$ is guaranteed.
For instance, in the hypothetical simplistic case where $f_r = \True$, and consider that \CB outputs a $\vsss$ with associated $c$ such that $f_c(x,y) = (y>2) \wedge (x<y)$ in $\ThZ$.
Then if the input is $x:5$, a candidate output $y:4$ of $D$ would result in
$f_r(5) \Into f_c(5,4) = (4>2) \wedge (5<4)$, which does not hold. Thus, $y:4$ must be overridden, so $\ShT$ will return a model of
$\exists y'. f_r(x) \Into (y'>2) \wedge (x<y')[x\leftarrow 5]$. One such possibility is $y':6$.
%

%
%

\begin{lemma} 
Let $\phiB$ be an equi-realizable abstraction of $\phiT$, \CB a controller for $\phiB$ and  \CT a corresponding controller for $\phiT$.
Also, consider input $\vxs$ and output $\vys$ by $D$ and $\vys'$ by \Provider of \CT.
If both $f_r(\vxs) \Into f_c(\vxs,\vys)$ and $f_r(\vxs) \Into f_c(\vxs,\vys')$ hold, then $\vys$ and $\vys'$ correspond to the same move of \CB.
\end{lemma}

\noindent In other words, if $\vys$ and $\vys'$ make the same literals true, a shield $\ShT$ based on $\CT$ will not override $\vys$; whereas, otherwise, it will output $\vys'$.

\begin{theorem} [Correctness of $D \cdot \ShT$ based on \CT] \label{thm:shieldsExist}
Let $D$ be an external controller and $\ShT$ be a shield constructed from a synthetized \CT from a specification $\phiT$. Then, $D \cdot \ShT$  is also a controller for $\phiT$.
\end{theorem}

In other words, given $\phiT$, if $\calT$ is decidable in the $\exists^*\forall^*$-fragment (a condition for the Boolean abstraction~\cite{RoCe23a}), then, leveraging \CT, a shield $\ShT$ conforming to $\phiT$ can be computed for every external $D$ whose input-output domain is $\calT$.


\subsubsection{More Flexible Shields from Winning Regions.}

\begin{figure*}[t!]
\centering
\hspace{0.5em}\begin{tabular}{c@{\hspace{0.5em}}c}
  \includegraphics[scale=0.46]{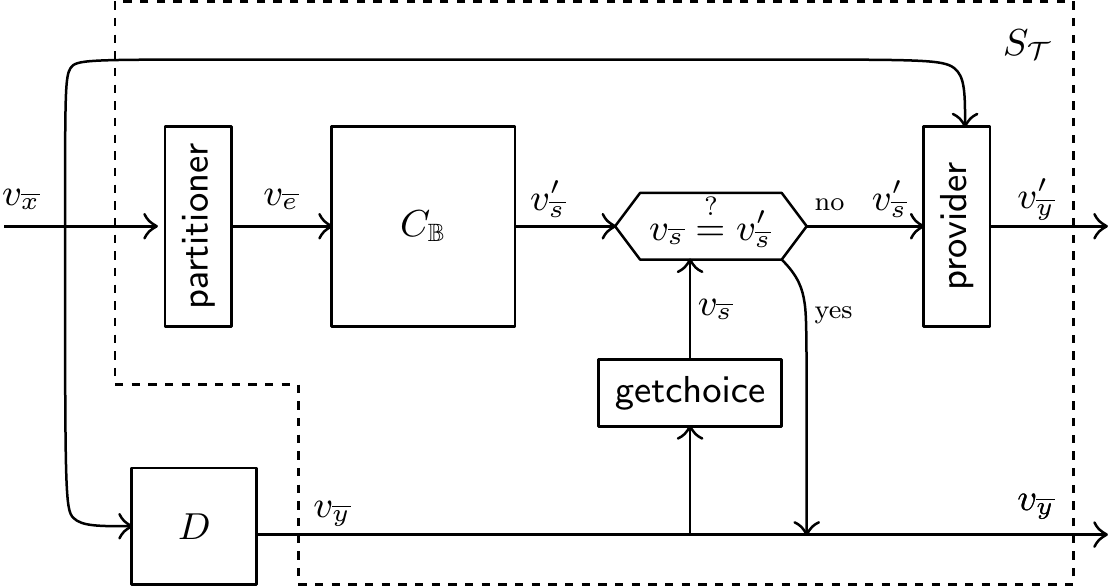} &
  \includegraphics[scale=0.46]{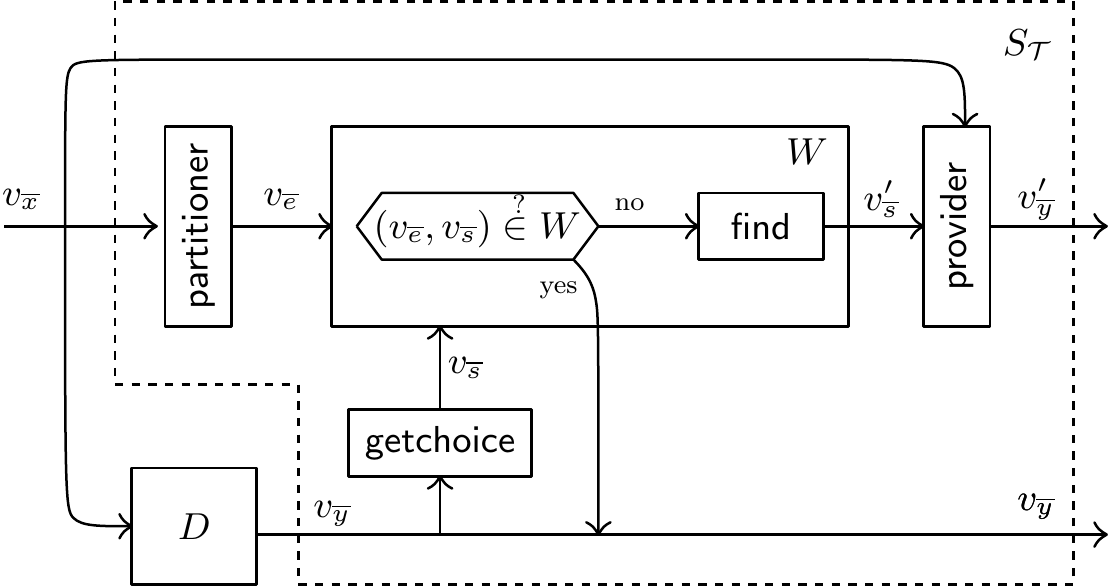} \\
  (a) Using \CT as $\ShT$. & (b) Using \WB as $\ShT$. 
\end{tabular}
\caption{Architectures for shields modulo theory.}
\label{fig:sts}
\end{figure*}


$\ShT$ based on \CT are sometimes too intrusive in the sense that they can cause unnecessary corrections, because some $(\vxs,\vys)$ may be labelled as a violation when they are not, only because it satisfies a different collection of literals than the ones chosen by the internal controller \CB.
Given input $\vxs$, the output $\vsss$ provided by \CB and the output $\vys$ suggest by $D$, $(\vxs,\vys)$ is considered to be a \emph{dangerous output}  (i.e., a potential violation of $\phiT$) whenever a $f_r(\vxs)$ holds but the $f_c(\vxs,\vys)$ of the $c$ associated to $\CB(\Partitioner(\vxs)) = \vsss$ does not hold.
Thus, we study now permissive shields $\ShT$ with more precise characterizations of unsafe actions.

%

The \emph{winning region} (\WR) 
is the most general characterization of the winning states and moves for a safety \LTL specification.
%
%
%
This second construction can be applied to any subset of \WR as long as every state has at least one successor for every input (for example, using the combination of a collection of controllers).
The process of computing $\ShT$ based on a given winning region $\WB:\tupleof{Q,I,T}$ is as follows (see Fig.~\ref{fig:sts}(b)) and maintains at each instant a set of current states $\Qnow\subseteq Q$ of \WB, starting from the initial set of states $I$.

\begin{enumerate}
    \item Compute $\phiB$ from $\phiT$ using abstraction, as before.
    \item We compute the winning region $\WB$ from $\phiB$ (e.g., using \cite{BrePeFraSan14}). 
    %
    \item The \Partitioner computes $\ves$ from the input values $\vxs$; and, given $\vys$ provided by $D$, we produce the unique $\vsss$ from $(\vxs,\vys)$
        or equivalently $c=\{s_i | l_i(\vxs,\vys) \text{ holds}\}$. This is computed by component \getchoice in Fig.~\ref{fig:sts}(b).
    \item From some of the current states $q\in \Qnow$ maintained, we decide whether $\delta(q,\ves)$ has some successor $(q',\vsss)$ or not. 
    If it does, output $\vys$.
     If it does not, choose a successor output $\vsss'$ such that there is $(q',\vsss')\in \delta(q,\ves)$ for some $q\in \Qnow$, and use
    a \Provider to produce an output $\vys'$ (this is implemented by $\find$ in Fig.~\ref{fig:sts}(b)).
    In either case, use appropriate valuation $\overline{v}$ of $\sss$ (i.e., $\vsss$ or $\vsss'$) to update $\Qnow$ to \(\{q'|(q',\overline{v})\in\delta(q,\ves) \text{ for some $q\in \Qnow$}\}\).
\end{enumerate}

Note that if the environment generates an input $\vxs$ (with corresponding $\ves$) and $D$ responds with an output $\vys$ 
such that $\getchoice(\vys)=\vsss$, 
if $\WB$ permits $(\ves,\vsss)$ then  the
suggested $\vys$  remains.
Otherwise, we use a component \find which computes an alternative $\vsss'$ such that $(\ves,\vsss')$ is in $\WB$ and a proper value $\vys$ that corresponds to $\vsss$ is output.
%
%
%
%
%
%
%
Also, note that $\find$ can always identify a value $\vsss'$ from any $q\in \Qnow$ such that there is a successor $(q',\vsss')\in \delta(q,\ves)$ by the definition of \WR.
Given two $\vsss$ and $\vsss'$ such that $(\ves,\vsss)\in \WB$ and $(\ves,\vsss')\in \WB$, the criteria (e.g., \emph{choice-logic} \cite{BeMaWo21}) to choose one is not relevant for this paper and we leave it for future work.
Finally, given $\vxs$, the \WB cannot distinguish two values $\vys$ and $\vys'$ if $\getchoice(\vxs,\vys)=\getchoice(\vxs,\vys')$.





\begin{theorem} [Correctness of $D \cdot \ShT$ based on WR] \label{thm:shieldsWR}
Let $\phiT$ be a safety specification and $\phiB$ an abstraction of $\phiT$.
Let  $\ShT$ be synthesized using the winning region $\WB$ of $\phiB$.
Then, $D \cdot \ShT$ satisfies $\phiT$, for any external controller $D$.
\end{theorem}

Thm.~\ref{thm:shieldsWR} holds because the valuations of the literals of any trace $\pi$ of $D\cdot \ShT$ is guaranteed to be a path in $\WB$ and therefore will satisfy $\phiB$. Hence, $\pi\models\phiT$.
The main practical difference between using a controller \CB and a winning region $W$, is expressed by the following lemma:

\begin{lemma}  \label{lemm:intervention}
Let $C_1$ be controller obtained using a controller \CB for $\phiB$ and $C_2$ obtained from the WR of $\phiB$. 
Let $D$ be an external controller.
Let $\pi$ be a trace obtained by $D\cdot C_1$ for some input $\pi_{\xs}$.
Then, $\pi$ is a trace of $(D\cdot C_1)\cdot C_2$ for $\pi_{\xs}$.
\end{lemma}

Lemma~\ref{lemm:intervention} essentially states that a WR controller is less intrusive than one based on a specific Boolean controller \CB.
%
%
However, in practice,  a WR is harder to interpret and slower to compute than \CB.
%
%
%
Note that we presented $\ShT$ based on \CT and $\ShT$ based on WR, but for classic \LTL shielding only the latter case makes sense.
This is because in classic \LTL using shield $S$ based on \CB is equivalent to forcing $D$ to coincide with $S$ on all time-steps, so $S$ could be directly used ignoring $D$.
However, in \LTLt there are potentially many different $\vys$ by $D$ that correspond to a
given Boolean $\vsss'$ by \CB, so even when this architecture forces to mimick $\CB$ in the Boolean domain, it is reasonable to use an external source $D$ of rich value output candidates.
Thus, not only we present an extension of shielding to $\LTLt$ (with the WR approach), but also this is the first work that is capable to synthetise shields leveraging all the power of classic synthesis (with the $\CT$ approach): e.g., we can use bounded synthesis \cite{SchFink07} instead of computing a WR.

\section{Permissive and Intrusive Shields} \label{subsec:optimality}

\subsubsection{Minimal and Feasible Reactions.}

Computing a Boolean Abstraction boils down to calculating the set \VR of valid reactions.
In order to speed up the calculation, optimizations are sometimes used (see \cite{RoCe23a}).
For instance, Ex.~\ref{ex:runningEx} shows a formula $\phiT$ that can be Booleanized into a $\phiB$ that contains three environment variables ($e_0$, $e_0^+$ and $e_1$) that correspond to the three valid reactions.
A closer inspection reveals that the valid reaction associated to $e_0^+$ (which represents $(x<1)$) is  ``subsumed'' by the reaction associated to $e_0$ (which represents $(x<2)$), which is witnessed by $e_0$ leaving to the system a strictly smaller set of choices than $e_0^+$.
%
%
In the game theoretic sense, an environment that plays $e_0^+$ is leaving the system strictly more options, so a resulting Booleanizaed formula $\phiB'$ that simply ignores $e_0^+$ only limits the environment from playing sub-optimal moves, and $\phiB'$ is also equi-realizable to $\phiT$.
From the equi-realizability point of view,  one can neglect valid reactions that provide strictly more choices than other valid reactions.
However, even though using $\phiB'$ does not compromise the correctness of the resulting $\ShT$ (using either method described before), the resulting $\ShT$  will again be more intrusive than if one uses $\phiB$, as some precision is lost (in particular for inputs $(x<1)$, which are treated  by $\phiB'$ as $(x<2)$).
We now rigorously formalize this intuition.


A reaction $r$ is \emph{above} another reaction $r'$ whenever $r'\subseteq r$.
%
%
Note that two reactions $r$ and $r'$ can be not comparable,
since neither contains the same or a strictly larger set of playable choices
than the other.
%
%
We now define two sets of reactions that are smaller than
$\VR$ but still guarantee that a $\phiB$
obtained from $\phiT$ by Boolean abstraction is
equi-realizable.
Recall that the set of valid reactions is $\VR=\{r|\exists \xs.f_r(\xs) \text{ is valid}\}$.

\begin{definition} [\MVR and Feasible]
 The set  of \emph{minimal valid reactions} is $\MVR = \{r \in \VR \;| \textit{ there is no } r' \in \VR
  \textit{ such that } r'\subseteq r\}$.
A set of reactions $R$ is a \textit{feasible} whenever
$\MVR \subseteq R \subseteq \VR$.
\end{definition}

That is, a set of reactions $R$ is feasible if all the reactions in
$R$ are valid and it contains at least all minimal reactions.
In order to see whether a set of reactions $R$ is \textit{below} $\VR$ or is indeed $\VR$, we need to check two properties.

\begin{definition}[Legitimacy and Strict Covering]
  \label{def:legitimacy}
  Let $R$ be a set of reactions:
  \begin{itemize}
  \item $R$ is \emph{legitimate} iff for all
  $r\in R$, $\exists \xs \suchThat f_r(\xs)$ is valid. 
  \item $R$ is a \emph{strict covering} iff $\forall\xs.\bigvee_{r\in R}f_r(\xs)$ is valid.
  \end{itemize}
\end{definition}


\noindent If $R$ is legitimate, then $R \subseteq \VR$.  If additionally $R$ is a strict covering, then $R = \VR$.
Strict covering implies that all possible moves of the environment are
covered, regardless of whether there are moves that a clever
environment will never play because these moves leave more power to
the system than better, alternative moves.
A \textit{non-strict} covering does not necessarily consider  all the
possible moves of the environment in the game, but it still considers at least all optimal environment moves.
%
%
A non-strict covering can be evaluated by checking that with regard to $R$, for all
$\overline{x}$, the disjunction of the \textbf{playable} choices (see $\bigwedge_{c\in r} \exists\ys. f_c$ of $f_r$ in Preliminaries) of its
reactions holds.

\begin{definition}[Covering]
  \label{def:nonStrict}
  
    The playable formula for a  reaction $r$ is defined as
$
  f^P_r(\xs) \DefinedAs \bigwedge_{c \in r} \exists \ys. f_c(\xs,\ys). 
$
  A set of reactions $R$ is covering if and only if 
  $\varphi_{cov}(R) = \forall \xs \suchThat \bigvee_{r \in
    R}f^P_r(\xs)$ is valid.
\end{definition}

Note that a playable formula removes from the characteristic formula $f_r$ the sub-formula
$(\bigwedge_{c\notin r} \forall\ys\neg f_c)$ that captures that the choices not in $r$ cannot be achieved by any $\vys$.
We can easily check whether a set of reactions $R$ is feasible, as follows.
First, $R$ must be legitimate. 
Second, if $\varphi_{cov}(R)$ is valid, then $R$ contains a subset
of valid reactions that makes $R$ covering in the
sense that it considers all the ``clever'' moves for the environment (considering that a move that leaves less playable choices to the system is more clever for the environment).
%

%

\begin{theorem} 
  \label{thm:equirealizabilityFVR}
  Let $R$ be a feasible set of reactions and let $\phiB = \phiT[l_i \leftarrow s_i] \wedge \phiEx(R)$ be the Booleanization of $\phiT$ using $R$.
  Then, $\phiT$ and $\phiB$ are equi-realizable.  
\end{theorem}

\subsubsection{Impact of Boolean Abstractions on Permissivity.}

To obtain an equi-realizable Boolean abstraction, it
is not necessary to consider \VR, and instead a feasible set is sufficient.
The computation of a feasible set is faster, and generates smaller $\phiB$ formulae~\cite{RoCe23a}.
However, there is a price to pay in terms of how permissive the shield is.
%
%
In practice, it regularly happens that the environment plays moves that are not optimal in the sense that other moves would leave less choice to the system.
\begin{example} \label{ex:shieldPermissivity}
    Recall Ex.~\ref{ex:runningEx}.  Note that a \CB from $\phiB$ can respond with $f_c$ of choices for $e_0$, $e_0^+$ and $e_1$.
    Now consider a $\phiB'$ that ignores $e_0^+$ and its eligible $f_c$. Thus, a \CB' from $\phiB'$ can only respond with $f_c$ for $e_0$ and $e_1$.
%
%
For the sake of the argument, consider the input $x:0$ forces to satisfy $f_{c_2}$ in \CB, or $f_{c_2}$ or $f_{c_3}$ in \CB'.
Thus, a candidate output  $y:2$ corresponds to holding $((x<10)\And\neg(y>9)\And\neg(y\leq x))$ which is exactly $f_{c_3}$ allowed by $e_0^+$, but not by $e_0$ (considered by $\phiB$ but not by $\phiB'$).
Therefore, the corresponding $\ShT$ using the winning region $\WB$ for $\phiB'$ would override the output candidate $y:2$ provided by $D$ (generating an output $\vys'$ such that $\vys'$ holds $f_{c_2}$ ; e.g., $y':-1$); i.e., it incorrectly interprets that candidate $\vys$ is dangerous, whereas $\ShT$ using the winning region $\WB$ for $\phiB$ would not override $\vys$.
\end{example}

%

%
The most permissive shield uses the \WR of the complete set of valid reactions \VR, but note that computing \WR and \VR is more expensive than \CB and \MVR.
However, for efficiency reasons the most permissive shield is usually not computed, either because (1) the abstraction algorithm for computing valid reactions computes a non-strict covering or because (2) the synthesis of \WR does not terminate (specially in liveness specifications).
Note that not only the cost of constructing $\ShT$ is relevant, but also other design decisions: if we want the policy of $D$ to dominate, then we need  $\ShT$ to be as permissive as possible, whereas if we want $\phiT$ to dominate (e.g., in specially critical tasks), then we want $\ShT$ to be intrusive.
Moreover, we can guarantee maximal permissivity computing \VR (and \WR), whereas we can also guarantee maximal intrusion computing \MVR (and \CB). Indeed, it is always possible to compute such \MVR.

\begin{theorem}
    \label{thm:maximalIntrusion}
    Maximally intrusive shield synthesis (i.e., \MVR and \CB) is decidable.    
\end{theorem}

In Ex.~\ref{ex:shieldPermissivity} the alternative $\phiB'$ that ignores the (playable) choice $e_0^+$ is exactly comparing $r_{e_0}=\{c_1, c_2\}$ and $r_{e_0^+}=\{c_1, c_2, c_3\}$ and constructing a feasible $R'$, where $r_{e_0^+} \notin R'$, since $r_{e_0} \subset r_{e_0^+}$ . 
In this case, $R'$ is an \MVR.
%
%
%


\section{Optimizations in $\LTLt$ Shields} \label{subsec:optimized}


\subsubsection{Shields Optimized in $\calT$.}

We know that, given $\vxs$, the \Provider component computes an output $\vys'$ from $c$ and $r$ such that $\vys'$ is a model of the formula
  $\psi = \exists \ys'. \big(f_r(\xs) \Into f_c(\ys',\xs)\big)[\xs\leftarrow\vxs]$.
Moreover, not only is $\psi$ guaranteed to be satisfiable (by Lemma.~\ref{lemm:satisfiability}), but usually has several models.
This implies that the engineer can select some $\vys'$ that are preferable 
over others, depending on different criteria represented by objective functions, such as 
$\psi_{f^+}$, where $f^+=\text{min}(|y-y'|)$, i.e., the objective function that minimizes the distance between $v_y$ by $D$ and $v_y'$ by $\ShT'$.
This is because these are not linear properties.
However, this optimization would be very relevant in the context of shielding, because, without loss of generality, it expresses that $\vys'$ is the safe correction by $\ShT$ closest to the unsafe $\vys$ by $D$. 
To solve this, we used maximum satisfiability, which adds \emph{soft constraints} $\mathcal{M} = \{\phi_1,\phi_2,...\}$ to $\psi$, such that $\psi(\mathcal{M}) = \exists \ys'. (f_c(\xs,\ys') ^{+}{}\bigwedge^{|\mathcal{M}|}_{i=0} \phi_i)[\xs \leftarrow \vxs]$, where $^{+}{}\bigwedge$ denotes a soft conjunction, meaning that the right-hand side is satisfied only if possible.
To better illustrate this, we use a single variable $y$, although this concept can also be extended to other notions of distance with multiple variables $\ys$ (e.g., Euclidean distance)
\footnote{In our
	extended paper~\cite{RoAmCoSaKa24} we include a formal discussion on how to 
	use 
	optimizations
depending on $\calT$.}.

\begin{example}    
Consider again Ex.~\ref{ex:shieldPermissivity} and
let $\mathcal{M}=\{(y'>5)\}$. Then $\psi(\mathcal{M}) = \exists y'. (f_c(x,y') ^{+}{}\wedge (y'>5))[x \leftarrow 0]$ does not return $y':2$, but some $y' \in [6,9]$. 
In addition, we use maximum satisfiability to express that $v_y'$ is the safe correction that is the closest to an unsafe $v_y$, for which we add $\phi(\xs,y,y') = \forall z. (f_c(\xs,z) \Into (|y'-y|<|z-y|))$ to $\mathcal{M}$ in $\psi(\mathcal{M})$, so that $\psi(\mathcal{M}) = (f_c(\xs,y) \wedge (y'>5) \wedge \phi(\xs,y,y'))[\xs \leftarrow \vxs, y \leftarrow v_y]$.
Thus, given a $y:4$ labelled as unsafe, $\psi(\mathcal{M})$ will not return an arbitrary model $y' \in [6,9]$, but the concrete $y':6$.
This converges in $\ThZ$.
\end{example}
%

Note that to enforce $\vys'$ closest to $\vys$, we need to additionally extended the \Provider component with input $\vys$.
Also, note that two soft constraints can be contrary to each other, in which case the engineer has to establish priorities using weights.
It is also important to note that using soft constraints does not compromise correctness, as any solution found by the solver that supports soft constraints will satisfy the hard constraints as well. Hence, the correctness of Thm.~\ref{thm:shieldsExist} and Thm.~\ref{thm:shieldsWR} remains in-place.
Note that not only $\ShT$ can minimize the distance to $D$ in arithmetic $\calT$, but in any $\calT$ for which engineers define a metric space.

\subsubsection{Permissive Optimization.}

We showed that a $\ShT$ can provide an output $\vys'$ that is the safe correction by closest to the unsafe $\vys$ by $D$.
Moreover, previously we showed that we can generate multiple strategies using \WR. 
Therefore, we can optimize $\vys'$ using combinations in the \WR; for instance, return $\vys'$ such that distance to $\vys$ is minimal and $\vys'$ has been chosen among the strategies represented by automata with less than $n$ states (i.e., using bounded synthesis with bound $n$).
We illustrate this in~\cite{RoAmCoSaKa24}.
    
%
\begin{theorem} \label{thm:wrDistance}
    Let $\calT=\ThZ$ and let a candidate output $\vys$ considered unsafe.
    Let a WR in an arbitrary state $q_k$ and an input $\vxs$, which yields a set $C=\{c_k,c_j,c_i,...\}$ of choices that are safe for the system. 
    We denote with $F=\{f_{c_k},f_{c_j},f{c_i}\}$ the set of characteristic choice functions.
    There is always a value $\vys'$ of $y'$ satisfying $f_r(\vxs) \Into f_c(y',\vxs)$, where $f_c \in F$, such that for $\vys''$ of $y''$ satisfying $f_r(\vxs) \Into f_c'(y'',\vxs)$, where $f_c' \in F$ and $f_c \neq f_c'$, then the distance from $\vys'$ to $\vys$ is smaller or equal than the distance from $\vys''$ to $\vys$.
\end{theorem}

Note that Thm.~\ref{thm:wrDistance} holds for other decidable $\calT$. 
This advantage is unique to $\ShT$ and it offers yet another permissivity layer to measure distance with respect to the policy of $D$.


 \section{Related Work and Conclusion}
\label{sec:related-work}

\subsubsection{Related Work.}
Classic shielding approaches~\cite{BlKoKoWa15,KogAlshBloemHumKogTopWang17,AlBloEh18,AvnBloeChattHenzKogPrang19} focus on properties expressed in Boolean LTL, and are incompatible for systems with richer-data domains: i.e., they need explicit \textbf{manual} discretization of the requirements.
In the other hand, we have a sound procedure that directly takes $\LTLt$ specifications (via~\cite{RoCe24})
which we adapted to shielding.
This is fundamentally new, and also that we can also optimize outputs with respect to erroneous candidates.

\subsubsection{Competing methods}

%
%
%

\begin{table}[t!]
\centering
  \begin{tabular}{c|cc|cccccc}
    \multirow{2}{*}{\textit{Req.}} & \multicolumn{2}{|c|}{(Wu et al.)} & \multicolumn{6}{|c}{\textit{Ours} } \\
        & $\mathbb{B}$ & $\calT$ & $\mathbb{B}$ & $\mathbb{B}'$ & $\calT$ & $\ThZ$ & $\calT_A$ & $\calT_B$\\
    \hline
    $\varphi_0$ & $30$ & $631$ & $19$ & $17$ & $171$ & $173$ & $183$ & $184$\\
    \hline
    $\varphi_1$ & $41$ & $590$ & 22  & 21 & 190 & 192 & 199 & 212\\
    \hline
    $\varphi_2$ & $80$ & $520$ & 24 & 20 & 214 & 194 & 214 & 199 \\
    \hline
    $\varphi_3$ & $45$ & $340$ & 22 & 18 & 106 & 105 & 120 & 105 \\
    \hline
    $\varphi_4$ & $50$ & $640$ & 21 & 17 & 118 & 122 & 120 & 117 \\
    \hline
    $\varphi_5$ & $80$ & $520$ & 21 & 16 & 124 & 124 & 127 & 156 \\
    \hline
    $\varphi_6$ & $37$ & $370$ & 16 & 14 & 141 & 156 & 148 & 151 \\
    \hline
    $\varphi_7$ & $49$ & $710$ & 20 & 20 & 163 & 159 & 165 & 173 \\
    \hline
    $\varphi_8$ & $45$ & $580$ & 21 & 15 & 133 & 134 & 143 & 189  \\
    \hline
    $\varphi_9$ & $18$ & $610$ & 14 & 11 & 170 & 179 & 172 & 203 \\
    \hline
    $\varphi_{10}$ & $50$ & $690$ & 25 & 19 & 173 & 168 & 182 & 177  \\
    \hline
    $\varphi_{11}$ & $31$ & $530$ & 21 & 17 & 177 & 178 & 191 & 177 \\
    \hline
    $\varphi_{12}$ & $57$ & $510$ & 29 & 22 & 156 & 155 & 179 & 182 \\
  \end{tabular}
  \caption{Comparison of our approach with \cite{WuMaDeWa19}, 
  measured in \protect$0.1 \mu s$ for $\mathbb{B}$ and in $\mu s$ for $\calT$, $\ThZ$, $\calT_A$ and $\calT_B$.}
  \label{tab:RealCyb}
\end{table}
%
%
%
The approach by \cite{WuMaDeWa19} is, to the best of our knowledge, the only previous successful attempt to compute rich-data shields.
We compared our work with theirs: for the $13$ different specifications from \cite{WuMaDeWa19}, we generated our shield and measured time for computing: (i) the Boolean step (clm. $\mathbb{B}$), and (ii) producing the final output (clm. $\calT$). We compared our results at Tab.~\ref{tab:RealCyb}.
%
%
For the first task, our approach, on average, requires less than $0.21\mu{}s$, while they take more than twice as long on average (with over $0.47\mu{}s$). Our approach was also more efficient in the second task, taking on average about $0.157ms$, whereas they take required  $0.557ms$ on average.
%
%
Moreover, note that our ($\mathbb{B}$) contains both detection and correction of candidate outputs, whereas their ($\mathbb{B}$) is only detection.
%
It is important to note another key advantage of our approach: ~\cite{WuMaDeWa19} necessarily generates a \WR, whereas we can also construct \CT (clm. $\mathbb{B}'$), performs considerably faster (and can operate leveraging highly matured techniques, e.g., bounded synthesis).
In addition to this, we \cite{WuMaDeWa19} presents a monolithic method that to be used with specifications containing only linear real arithmetic $\ThZ$. Hence, if we slightly modify $\calT$ to $\ThZ$, then the method in~\cite{WuMaDeWa19} cannot provide an appropriate shield, whereas we do (clm. $\ThZ$).
Indeed, our method encodes \emph{any} arbitrary 
$\exists^*\forall^*$ decidable fragment of $\calT$ (e.g., non-linear arithmetic or the array property fragment~\cite{BrMaSi06}).
%
Additionally, we can optimize the outputs of $\ShT$ with respect to different criteria, 
such as returning the smallest/greatest safe output  (clm. $\calT_A$) and the output closest to the candidate (clm. $\calT_B$). 
For this, we used Z3 with optimization \cite{BjoDunFleck15}.
Tab.~\ref{tab:RealCyb} shows that our approach can manage both rich data and 
temporal dynamics. 
%

 \subsubsection{Conclusion.}
%
In this work, we present the first general methods for shielding properties encoded in $\LTLt$. 
These allows engineers to guarantee the safe behavior of DRL agents in complex, reactive environments.
Specifically, we demonstrate how shields can be computed 
from a controller and from a winning region.
This is not simply a direct application of synthesis, instead, we can guarantee maximally and minimally permissive shields, as well as shields optimized with respect to objective functions in arbitrary decidable $\calT$.
%
%
We also empirically demonstrate its applicability, 
(see \cite{CoAmRoKaSaFo24} for an actual application).
%

A next step is to adapt shields modulo theories to probabilistic settings \cite{PraKoPoBlo21,CarJanJunTop23}
and planning \cite{CaBienMcil19}. 
%
%
%
We also want to leverage recent results on finite-trace $\LTLt$~\cite{GeGiGiWi23}, 
to model more expressive shields. 

\section{Acknowledgments}
	The work of Rodríguez and Sánchez was funded in part by PRODIGY Project 
	(TED2021-132464B-I00) — funded
	by MCIN/AEI/10.13039/501100011033/ and the European Union 
	NextGenerationEU/PRTR — by
	the DECO Project (PID2022-138072OB-I00) — funded by 
	MCIN/AEI/10.13039/501100011033 and
	by the ESF, as well as by a research grant from Nomadic Labs and the Tezos 
	Foundation.
	This work was partially funded by the European Union (ERC, VeriDeL, 
	101112713). 
	Views and opinions expressed are however those of the author(s) only and do 
	not 
	necessarily reflect those of the European Union or the European Research 
	Council Executive Agency. Neither the European Union nor the granting 
	authority 
	can be held responsible for them.
	The work of Amir was further supported by a scholarship from the Clore 
	Israel 
	Foundation as well as a Rothschild Fellowship granted by the Yad Hanadiv 
	Foundation.

\bibliography{aaai25}

\end{document}